\documentclass[preprint,12pt]{elsarticle}
\usepackage[T1]{fontenc}
\usepackage{lmodern}
\usepackage{amsmath,amssymb,amsthm,mathtools}
\usepackage{microtype,booktabs,array,enumitem}
\usepackage{tikz}
\usetikzlibrary{arrows.meta,positioning,calc,fit,backgrounds}
\usepackage{needspace}
\usepackage{xurl}
\usepackage{hyperref}
\hypersetup{colorlinks=true,linkcolor=blue!45!black,citecolor=blue!45!black,urlcolor=blue!45!black,pdftitle={PSPACE-completeness of the Hunters and Rabbit game},pdfauthor={Walid BEN-AMEUR and Alessandro Maddaloni}}
\definecolor{stateblue}{RGB}{31,86,130}
\definecolor{guardgray}{RGB}{239,241,244}
\definecolor{badred}{RGB}{165,45,45}
\definecolor{testgreen}{RGB}{38,108,76}
\newtheorem{theorem}{Theorem}
\newtheorem{lemma}[theorem]{Lemma}
\newtheorem{proposition}[theorem]{Proposition}
\newtheorem{corollary}[theorem]{Corollary}
\newcommand{\PSPACE}{\mathsf{PSPACE}}
\newcommand{\NPSPACE}{\mathsf{NPSPACE}}
\newcommand{\q}{\boldsymbol q}
\newcommand{\s}{\boldsymbol s}
\newcommand{\cmd}{\Gamma}
\newcommand{\inv}[1]{\overline{#1}}
\newcommand{\test}{\#}
\setlist{itemsep=3pt,topsep=5pt}
\newcommand{\AItool}{ChatGPT (OpenAI)}
\journal{Information and Computation}
\date{13 October 2026}
\biboptions{numbers,sort&compress}
\begin{document}
\begin{frontmatter}
\title{PSPACE-completeness of the Hunters and Rabbit game}
\author{Walid Ben-Ameur}
\ead{walid.benameur@telecom-sudparis.eu}
\author{Alessandro Maddaloni\corref{corresponding}}
\ead{alessandro.maddaloni@telecom-sudparis.eu}
\cortext[corresponding]{Corresponding author.}
\affiliation{organization={SAMOVAR, T\'el\'ecom SudParis, Institut Polytechnique de Paris},
  addressline={19 place Marguerite Perey},
  city={Palaiseau},postcode={91120},country={France}}

\begin{abstract}
In the Hunters and Rabbit game, hunters shoot at a set of vertices at each round, while an invisible rabbit must move along an edge. We prove that deciding whether a given number of hunters can guarantee capture is PSPACE-complete, even on simple bipartite graphs. The reduction uses intersection nonemptiness for inverse finite automata: one common word must put all automata in accepting states simultaneously. A clique acting as a safe haven and two families of guards force every winning strategy to simulate common transitions. Existing transformations remove the restriction on the rabbit's initial positions and yield PSPACE-completeness for the directed game with one hunter, even on oriented graphs with positive minimum outdegree. Adjusting the guard sizes also proves PSPACE-completeness of the firefighting game, in which the rabbit may remain at its current vertex.
\end{abstract}
\end{frontmatter}

\section{Introduction}
The Hunters and Rabbit game is played on a finite graph by a team of hunters and an invisible rabbit. At each round, the hunters shoot at a set of vertices, with at most one shot per hunter. A surviving rabbit must then move along an edge. The hunters may choose arbitrary vertices at successive rounds, and receive no information about the rabbit's position before capture. A winning strategy guarantees capture for every initial position and every subsequent rabbit walk. The minimum number of hunters admitting such a strategy is the \emph{hunting number} $h(G)$. We study the decision problem that asks, given a graph $G$ and an integer $k$, whether $h(G)\leq k$.

Early studies of the one-hunter case appeared under other names. Britnell and Wildon~\cite{BritnellWildon} considered a prince searching for an invisible princess who changes rooms at every step, characterized the graphs on which capture is possible, and determined the optimal capture time. Haslegrave~\cite{Haslegrave} independently studied the corresponding evasion game and obtained a characterization and optimal strategies. These results provide a polynomial-time solution of the one-hunter problem on undirected graphs.

Abramovskaya, Fomin, Golovach and Pilipczuk~\cite{AFGP} developed the game with an arbitrary number of hunters. They proved the general bound $h(G)\leq\operatorname{pw}(G)+1$, where $\operatorname{pw}$ denotes pathwidth, determined the hunting number of an $r\times s$ grid as $\lfloor\min\{r,s\}/2\rfloor+1$, and established an $O(\log n)$ upper bound for $n$-vertex trees together with a nearly matching lower bound. Gruslys and M\'eroueh~\cite{GruslysMeroueh} obtained a logarithmic lower bound on trees, showing that this order of growth is best possible. Bolkema and Groothuis~\cite{BolkemaGroothuis} determined the hunting number of hypercubes and related the game to isoperimetric nesting properties.

Monotonicity provides a further connection with graph decompositions. Dissaux, Fioravantes, Gahlawat and Nisse~\cite{DFGN} introduced the monotone hunting number $mh(G)$ and proved $\operatorname{pw}(G)\leq mh(G)\leq\operatorname{pw}(G)+1$. They also established polynomial-time algorithms for this parameter on several graph classes and fixed-parameter tractability of the ordinary hunting number with respect to vertex cover number. Their examples show that allowing recontamination can be essential: for every $r\geq3$, there is a tree $T$ with $h(T)=2$ and $mh(T)=r$. Thus an argument based solely on monotone strategies does not capture the complexity of the ordinary game.

Ben-Ameur, Gahlawat and Maddaloni~\cite{BGMCOCOON} proved NP-hardness of computing the hunting number, even for simple bipartite graphs. Their extended journal version~\cite{BGM} proves that approximating $h(G)$ within additive error $O(n^{1-\varepsilon})$ is NP-hard for every fixed $\varepsilon>0$, where $n=|V(G)|$, and also establishes NP-completeness for prescribed polynomially bounded time limits of at least two rounds. The unrestricted problem belongs to PSPACE, but their concluding remarks leave open whether it is PSPACE-complete.

For directed graphs, Ben-Ameur and Maddaloni~\cite{BM} proved NP-hardness already for one hunter. They also gave a polynomial-time algorithm for the one-hunter problem on tournaments and established fixed-parameter tractability results under structural parameterizations. Their earlier work~\cite{BMPathwidth} gives $h(D)\leq\operatorname{dpw}(D)+1$, where $\operatorname{dpw}$ denotes directed pathwidth. For monotone strategies, whose territories are nonincreasing under inclusion, the minimum number of hunters lies between $\operatorname{dpw}(D)$ and $\operatorname{dpw}(D)+1$.

Allowing the invisible fugitive to remain at its current vertex gives the inspection game studied by Bernshteyn and Lee~\cite{BernshteynLee}. They call the minimum number of vertices inspected per round the \emph{inspection number}, establish bounds in terms of graph structure, and study its behaviour under subdivisions. The same game has a firefighting interpretation: all vertices are initially burning, firefighters extinguish a bounded number of vertices per round, and fire then spreads from the remaining burning vertices, possibly reigniting vertices just extinguished. Althoetmar, Schade and Sch\"urenberg~\cite{Firefighting} proved NP-hardness for this model and constructed instances requiring superpolynomially long shortest winning strategies, with a transformation transferring these results to Hunters and Rabbit. The absence of a polynomial bound on strategy length explains why NP membership cannot be inferred simply by taking a complete shooting sequence as a certificate.

The precise movement and information rules matter when comparing these results with the broader literature; see Fomin and Thilikos~\cite{FominThilikos} for a survey. In the classical Cops and Robbers game, studied by Nowakowski and Winkler~\cite{NowakowskiWinkler} and Quilliot~\cite{Quilliot}, the robber is visible and both players move along graph edges, with the option of staying at their current vertices. Kinnersley~\cite{Kinnersley} proved that deciding whether $k$ cops suffice is EXPTIME-complete when $k$ is part of the input. The hunters in our game can instead relocate their shots arbitrarily, while the rabbit is invisible and must move after every unsuccessful round.

Other graph-searching models allow searchers to be placed or removed at arbitrary vertices, while the fugitive has unbounded speed: it can traverse an arbitrarily long path that avoids the vertices blocked by searchers. This path-constrained motion differs from unrestricted relocation. Such models connect pursuit games with treewidth and pathwidth~\cite{FominThilikos}. In particular, Seymour and Thomas~\cite{SeymourThomas} characterized treewidth through a game with a visible fugitive of unbounded speed. Monotonicity theorems, including LaPaugh's result that recontamination does not help in classical graph searching~\cite{LaPaugh}, are central to this connection.

PSPACE-completeness is known for some related models. Akhoondian Amiri, Kaiser, Kreutzer, Rabinovich and Siebertz~\cite{AmiriEtAl} proved that deciding DAG-width is PSPACE-complete. In its game formulation, cops may relocate arbitrarily, the visible robber follows directed paths avoiding the cops that remain in place, and the cops must use a robber-monotone strategy: a vertex made unreachable must not become reachable again. Kreutzer and Ordyniak~\cite{KreutzerOrdyniak} proved PSPACE-completeness for the invisible domination game even with two cops. There the robber has unbounded speed and a cop controls its closed neighborhood, rather than only its own vertex. These results concern different movement, capture or monotonicity rules and do not imply the corresponding classification for Hunters and Rabbit.

We settle the PSPACE-completeness question for the ordinary Hunters and Rabbit game, with the number of hunters given as part of the input.
\begin{theorem}\label{thm:main}
Deciding whether $h(G)\leq k$ is PSPACE-complete, even when $G$ is a simple bipartite graph.
\end{theorem}
The proof reduces intersection nonemptiness for inverse finite automata, a PSPACE-complete problem~\cite{BirgetMargolis}, to the existence of a winning hunting strategy. The construction simulates a common word in all automata; inverse transitions allow this simulation to use an undirected graph.

We also obtain PSPACE-completeness for the directed game with one hunter, even on oriented graphs of positive minimum outdegree (Corollary~\ref{cor:directed}), and for the firefighting game, equivalently for deciding whether the inspection number is at most a given integer (Theorem~\ref{thm:firefighting}).

Section~2 gives the preliminaries and the polynomial-space upper bound. Section~3 proves the hardness results.

\section{Preliminaries and membership in PSPACE}
Graphs may contain loops in intermediate constructions, but have no parallel edges; a simple graph has no loops. For a graph $G=(V,E)$ and $v\in V$, define its neighborhood by $N_G(v)=\{u\in V:uv\in E\}$. Thus $v\in N_G(v)$ exactly when $v$ has a loop, which allows the rabbit to remain at $v$. For $R\subseteq V$, write $N_G(R)=\bigcup_{v\in R}N_G(v)$ and $N_G[R]=R\cup N_G(R)$, omitting the subscript when unambiguous. A strategy using $k$ hunters is a sequence $(W_t)_{t\geq1}$ with $W_t\subseteq V$ and $|W_t|\leq k$. The hunters know that the rabbit starts in a specified set $S\subseteq V$, but not its exact initial vertex. Its possible positions immediately after the shots at round $t$ form the \emph{rabbit territory} $R_t$, where
\begin{equation}\label{eq:territory}
R_1=S\setminus W_1,\qquad R_{t+1}=N_G(R_t)\setminus W_{t+1}\quad(t\geq1).
\end{equation}
By induction, $R_t$ is exactly the set of positions attainable after surviving the first $t$ shots.

The strategy is winning if $R_T=\varnothing$ for some finite $T$. We denote the minimum required number of hunters by $h_S(G)$; thus $h(G)=h_V(G)$. The input to the restricted-start decision problem is $(G,S,k)$, and the question is whether $h_S(G)\leq k$. For a digraph $D=(V,A)$, put $N_D^+(v)=\{u\in V:(v,u)\in A\}$ and $N_D^+(R)=\bigcup_{v\in R}N_D^+(v)$, and use $N_D^+(R)$ in place of $N_G(R)$. A rabbit unable to move loses. All graphs used below for hardness have no isolated vertices, and all resulting digraphs have positive minimum outdegree.

The territory accounts for all surviving rabbit walks simultaneously, so an empty territory certifies capture against every walk. Since the hunters receive no observations before capture, their strategy is a sequence of shot sets.

We recall the standard territory argument, including the exponential bound on capture time from~\cite[Proposition~1]{AFGP}.
\begin{lemma}\label{lem:space}
Deciding whether $h_S(G)\leq k$ belongs to $\PSPACE$, for undirected graphs and digraphs.
\end{lemma}
\begin{proof}
Put $n=|V|$. If $S=\varnothing$, accept immediately. Otherwise, let $W_1,\ldots,W_T$ be a shortest winning sequence, if one exists. Its nonempty territories $R_1,\ldots,R_{T-1}$ are pairwise distinct: if $R_a=R_b$ with $a<b<T$, deleting $W_{a+1},\ldots,W_b$ preserves the winning suffix, since future evolution depends only on the current territory. Hence $T-1\leq2^n-1$, and $T\leq2^n$.

A nondeterministic machine chooses one set $W_t$ of size at most $k$ at each round and updates the territory using~\eqref{eq:territory}. It accepts as soon as the territory is empty, and rejects the current branch if this has not occurred after $2^n$ rounds. The current territory, shot set and temporary neighborhood each require $n$ bits; the binary round counter requires $n+1$ bits. A neighborhood is computed by scanning the input edges, and previous shot sets are discarded. Thus every branch halts and uses polynomial space on every input. An accepting branch exists exactly when a winning sequence exists, by the length bound. Therefore the problem is in $\NPSPACE=\PSPACE$ by Savitch's theorem~\cite{Savitch}. The same argument uses the out-neighborhood for digraphs.
\end{proof}

\Needspace{10\baselineskip}
Let $\Sigma=\{a,b,\inv a,\inv b\}$, with $\inv{\inv a}=a$ and $\inv{\inv b}=b$.
The bar denotes a formal inverse symbol which is part of the input alphabet. An \emph{inverse finite automaton} is a tuple $A_i=(Q_i,\Sigma,\delta_i,s_i,F_i)$, where $Q_i$ is a finite state set, $s_i\in Q_i$ is the initial state, $F_i\subseteq Q_i$ is the set of accepting states, and $\delta_i:Q_i\times\Sigma\rightharpoonup Q_i$ is a partial transition function satisfying
\begin{equation}\label{eq:inverse}
\delta_i(q,c)=r\quad\Longleftrightarrow\quad\delta_i(r,\inv c)=q
\qquad(q,r\in Q_i,\ c\in\Sigma).
\end{equation}
In particular, for each symbol $c$, the map $\delta_i(\cdot,c)$ is injective on its domain, and its inverse is $\delta_i(\cdot,\inv c)$. A transition is \emph{partial}: a symbol may be unavailable at a state. If an unavailable transition is requested, the word is not accepted.

For a finite word $w=c_1\cdots c_\ell$, begin at $q_i^0=s_i$ and recursively set $q_i^t=\delta_i(q_i^{t-1},c_t)$. The word belongs to $L(A_i)$ exactly when all these transitions are defined and $q_i^\ell\in F_i$. The empty word $\varepsilon$ is accepted exactly when $s_i\in F_i$.

\smallskip
\begingroup
\setlength{\fboxsep}{8pt}
\noindent\fbox{%
\begin{minipage}{\dimexpr\linewidth-2\fboxsep-2\fboxrule\relax}
\textbf{Inverse-automata intersection nonemptiness.}\smallskip\\
\textbf{Input:} a list $A_1,\ldots,A_m$ of inverse finite automata over the common alphabet $\Sigma$, where $m\geq1$ is part of the input.\\[2pt]
\textbf{Question:} does there exist one finite word $w\in\Sigma^*$ such that $\delta_i(s_i,w)$ is defined and belongs to $F_i$ for every $i=1,\ldots,m$?
Equivalently, is $\bigcap_{i=1}^m L(A_i)\neq\varnothing$?
\end{minipage}%
}
\endgroup
\smallskip

The same word is read by every automaton, one symbol at a time. Thus we seek one time $\ell$ when \emph{all} current states are accepting. It is insufficient that each automaton can accept some word of its own, or that each visits an accepting state at a different prefix of a common word.

It is useful to view the problem in the product state space $Q_1\times\cdots\times Q_m$. Write $\s=(s_1,\ldots,s_m)$ and, for $\q=(q_1,\ldots,q_m)$, define $\delta_c(\q)=(\delta_1(q_1,c),\ldots,\delta_m(q_m,c))$ only when every component transition is defined. A tuple is \emph{reachable} if some common word takes $\s$ to it, and \emph{accepting} if it belongs to $F_1\times\cdots\times F_m$. We seek a reachable accepting tuple. The product may have exponentially many tuples, although one tuple has a polynomial-size encoding.

Birget and Margolis~\cite[Theorem~3.1]{BirgetMargolis} prove PSPACE-completeness of intersection \emph{emptiness} for inverse finite automata over the fixed alphabet $\Sigma=\{a,b,\inv a,\inv b\}$, even with a single accepting state per automaton. Since PSPACE is closed under complement, intersection \emph{nonemptiness} is also PSPACE-complete. We allow arbitrary $F_i$ in the notation; hardness already holds when each $F_i$ is a singleton.

\section{PSPACE-completeness}
\subsection{Construction}
We first restrict the rabbit's possible initial positions. A simulated configuration consists of one state from each automaton. Two copies of the state sets and undirected transition paths encode a common letter: traversing a path backwards reads the inverse letter. Guard sets, together with a clique in which the rabbit can survive indefinitely, force complete state tuples and a common command to be preserved until an accepting tuple is tested. Test vertices allow capture precisely when all component states are accepting.

Fix inverse automata $A_1,\ldots,A_m$ as above. We construct a simple undirected graph $G$ and an initial subset of vertices $S$. All vertices and sets introduced below are distinct.

For every $i$ and $q\in Q_i$, introduce two \emph{state vertices} $x^0_{i,q}$ and $x^1_{i,q}$. The superscript is called the \emph{side}. For $j\in\{0,1\}$, put
\begin{equation}\label{eq:tuple}
X^j(\q)=\{x^j_{i,q_i}:1\leq i\leq m\}.
\end{equation}
For each $i$ and $c\in\Sigma$, introduce \emph{transition vertices} as follows.
\begin{itemize}
\item If $\delta_i(q,c)=r$ is defined, introduce a vertex $y_{i,c,q}$ adjacent to $x^0_{i,q}$ and $x^1_{i,r}$; its command is $c$.
\item If $\delta_i(q,c)$ is undefined, introduce a \emph{bad transition vertex} $y^{0,\mathrm{bad}}_{i,c,q}$ of command $c$, with sole state neighbor $x^0_{i,q}$.
\item For every $r$ outside the image of $\delta_i(\cdot,c)$, introduce a bad transition vertex $y^{1,\mathrm{bad}}_{i,c,r}$ of command $c$, with sole state neighbor $x^1_{i,r}$.
\end{itemize}
For each command $c\in\Sigma$, every state vertex has exactly one transition neighbor, good or bad. On side~$0$, the neighbor of $x^0_{i,q}$ is $y_{i,c,q}$ if $\delta_i(q,c)$ is defined, and $y^{0,\mathrm{bad}}_{i,c,q}$ otherwise. On side~$1$, if $r=\delta_i(q,\inv c)$ is defined, then~\eqref{eq:inverse} gives $\delta_i(r,c)=q$, and $r$ is the unique predecessor of $q$ under $c$. The neighbor of $x^1_{i,q}$ is then $y_{i,c,r}$; otherwise, it is $y^{1,\mathrm{bad}}_{i,c,q}$. Every good transition vertex has exactly one state neighbor on each side, and every bad transition vertex has exactly one state neighbor in total. For a defined transition $\delta_i(q,c)=q'$, the good transition vertex $y_{i,c,q}$ lies on the path $x^0_{i,q}-y_{i,c,q}-x^1_{i,q'}$. Traversing this path from side~$0$ to side~$1$ simulates the automaton transition $q\xrightarrow{c}q'$. Traversing it from side~$1$ to side~$0$ simulates the inverse transition $q'\xrightarrow{\inv c}q$, since $\delta_i(q',\inv c)=q$.

Introduce a new command symbol $\test\notin\Sigma$ and put $\cmd=\Sigma\cup\{\test\}$. This is a command of the \emph{graph construction}, not an input letter of the automata; no transition $\delta_i(\cdot,\test)$ is defined. For $c\in\Sigma$, the vertices of command $c$ are the good and bad transition vertices indexed by $c$. The vertices of command $\test$ are all test vertices. Each state vertex $x^j_{i,q}$ receives a private \emph{test vertex} $y^j_{i,\test,q}$ of command $\test$, with sole state neighbor $x^j_{i,q}$. Mark it bad exactly when $q\notin F_i$. Each state vertex now has four transition neighbors and one test neighbor, one for each command in $\cmd$. The test command $\test$ will be used to check whether all automata are simultaneously in accepting states.

Introduce a clique $U=K_{9m+2}$, two independent guard sets $P_0,P_1$ of size $5m$, and an independent guard set $C_c$ of size $8m$ for each $c\in\cmd$. Join every state vertex on side $j$ to all of $P_j$, and every transition or test vertex of command $c$ to all of $C_c$. Join every guard, every bad transition vertex and every bad test vertex to all of $U$. Finally introduce an independent set $Z$ of size $9m$, join it completely to $U$, and set
\begin{equation}\label{eq:start}
S=Z\cup X^0(\s).
\end{equation}
There are no other edges. The transition and test vertices are specific to their automaton, state and command indices; different automata never share them. The sets $U,P_0,P_1$ and each $C_c$ are shared across the construction. Figure~\ref{fig:construction} displays every edge type. All these objects are ordinary vertices or sets of vertices of the same simple undirected graph.

\begin{figure}[p]
\centering
\resizebox{.96\linewidth}{!}{%
\begin{tikzpicture}[x=1cm,y=1cm,font=\small,
 box/.style={draw,rounded corners=1pt,fill=guardgray,minimum height=6mm,inner sep=4pt,align=center},
 state/.style={circle,draw=stateblue,fill=stateblue!5,minimum size=8mm,inner sep=1pt},
 vertex/.style={circle,draw,minimum size=8mm,inner sep=1pt},
 bad/.style={vertex,draw=badred,fill=badred!8,font=\small},
 edgejoin/.style={line width=.65pt},
 danger/.style={edgejoin,draw=badred},
 title/.style={anchor=west,font=\small\bfseries}]
\path[use as bounding box](-.15,-18.35) rectangle(15.5,.4);
\node[title] at (0,0) {(a) Common guard structure and forced initialization};
\node[box] (u) at (7.5,-.85) {$U=K_{9m+2}$};
\node[box] (z) at (13,-.85) {$Z$\quad $9m$};
\draw[edgejoin] (u)--(z);
\foreach \name/\lab/\xx in {p0/P_0/0.8,ca/C_a/3.1,cb/C_b/5.3,cai/C_{\inv a}/7.5,cbi/C_{\inv b}/9.7,ch/C_{\test}/11.9,p1/P_1/14.2}{
 \node[box] (\name) at (\xx,-2.1) {$\lab$};
 \draw[edgejoin] (u.south)--(\name.north);
}
\node[font=\scriptsize] at (.8,-2.65) {$5m$};
\node[font=\scriptsize] at (14.2,-2.65) {$5m$};
\node[font=\scriptsize] at (7.5,-2.7) {Each of the five command guard sets has $8m$ vertices.};
\node[anchor=west,font=\small] at (0,-3.2) {$S=Z\cup X^0(\s)$; the first shot is forced: $W_1=Z$, $R_1=X^0(\s)$.};

\node[title] at (0,-4.0) {(b) Defined transition: $\delta_i(q,c)=r$, with $c\in\Sigma$};
\node[box] (ub) at (7.5,-4.75) {$U$};
\node[box] (pb0) at (2.1,-5.75) {$P_0$};
\node[box] (cbc) at (7.5,-5.75) {$C_c$};
\node[box] (pb1) at (12.9,-5.75) {$P_1$};
\node[state] (xb0) at (2.1,-6.85) {$x^0_{i,q}$};
\node[vertex] (ppb) at (7.5,-6.85) {$y_{i,c,q}$};
\node[state] (xb1) at (12.9,-6.85) {$x^1_{i,r}$};
\foreach \a/\b in {ub/pb0,ub/cbc,ub/pb1,pb0/xb0,cbc/ppb,pb1/xb1,xb0/ppb,ppb/xb1}{\draw[edgejoin](\a)--(\b);}
\node[font=\scriptsize] at (7.5,-7.55) {Side $0\to1$ reads $c$; side $1\to0$ reads $\inv c$. The transition vertex has no edge to $U$.};

\node[title] at (0,-8.45) {(c) Undefined successor};
\node[title] at (8,-8.45) {(d) Missing predecessor};
\node[anchor=west,font=\small] at (0,-9.0) {$\delta_i(q,c)$ undefined};
\node[anchor=west,font=\small] at (8,-9.0) {$r\notin\operatorname{im}\delta_i(\cdot,c)$};
\node[box] (uc) at (3.4,-9.65) {$U$};
\node[box] (pc0) at (1,-10.65) {$P_0$};
\node[box] (ccc) at (5.4,-10.65) {$C_c$};
\node[state] (xc) at (1,-11.9) {$x^0_{i,q}$};
\node[bad] (pc) at (5.4,-11.9) {$y^{0,\mathrm{bad}}_{i,c,q}$};
\foreach \a/\b in {uc/pc0,uc/ccc,pc0/xc,ccc/pc,xc/pc}{\draw[edgejoin](\a)--(\b);}
\draw[danger](uc.south)--(pc.north west);
\node[box] (ud) at (11.4,-9.65) {$U$};
\node[box] (pd1) at (9,-10.65) {$P_1$};
\node[box] (ccd) at (13.4,-10.65) {$C_c$};
\node[state] (xd) at (9,-11.9) {$x^1_{i,r}$};
\node[bad] (pd) at (13.4,-11.9) {$y^{1,\mathrm{bad}}_{i,c,r}$};
\foreach \a/\b in {ud/pd1,ud/ccd,pd1/xd,ccd/pd,xd/pd}{\draw[edgejoin](\a)--(\b);}
\draw[danger](ud.south)--(pd.north west);
\node[font=\scriptsize] at (7.5,-12.8) {Each bad transition vertex has one state neighbor, its command guards, and all of $U$.};

\node[title] at (0,-13.55) {(e) Private test vertex: one for every state vertex};
\node[box] (ue) at (7.5,-14.3) {$U$};
\node[box] (pe) at (3,-15.3) {$P_j$};
\node[box] (ce) at (11.5,-15.3) {$C_{\test}$};
\node[state] (xe) at (3,-16.5) {$x^j_{i,q}$};
\node[vertex,draw=testgreen,fill=testgreen!6] (ze) at (11.5,-16.5) {$y^j_{i,\test,q}$};
\foreach \a/\b in {ue/pe,ue/ce,pe/xe,ce/ze,xe/ze}{\draw[edgejoin](\a)--(\b);}
\draw[danger,dashed](ue.south)--node[pos=.64,below left,font=\scriptsize,text=badred]{$q\notin F_i$ only}(ze.north west);
\node[anchor=west,font=\small] at (0,-17.4) {The test vertex is good if $q\in F_i$, and bad (joined to $U$) if $q\notin F_i$.};
\node[anchor=west,font=\small] at (0,-18.05) {$\test\notin\Sigma$: the test is a graph command, not an automaton input symbol.};
\end{tikzpicture}%
}
\caption{All edge types of the construction. Panel (a) shows the shared sets $U,P_j,C_c$; panels (b)--(e) show local connections. Circles denote vertices and boxes denote vertex sets; lines to boxes represent complete adjacency. The set $U$ induces a clique; the other boxed sets are independent. Transition and test vertices are specific to their displayed indices. Red lines join bad vertices to $U$; the dashed line is present exactly when $q\notin F_i$. All edges are undirected.}\label{fig:construction}
\end{figure}

We now consider the Hunters and Rabbit game on $G$, with the rabbit initially restricted to $S=Z\cup X^0(\s)$ and at most $k=9m$ hunters. The following lemmas show how the shot budget and guard sets force the simulation.
\begin{lemma}\label{lem:forbidden}
In a winning strategy using at most $k$ hunters, no territory contains a vertex of $U$, a guard, a bad transition vertex, a bad test vertex, or a vertex of $Z$. Moreover, no territory contains state vertices on both sides, transition/test vertices of two different commands, or both a state vertex and a transition or test vertex.
\end{lemma}
\begin{proof}
If $R_t$ contains a vertex of $U$, then $|N(R_t)\cap U|\geq k+1$, so a vertex of $U$ survives the next shot. Inductively, capture is impossible. Every guard, bad transition vertex, bad test vertex, and vertex of $Z$ is adjacent to all $k+2$ vertices of $U$; if one survives a shot, a vertex of $U$ survives the next one. These vertices are therefore forbidden in every territory of a winning strategy.

State vertices on both sides would make all of $P_0\cup P_1$ reachable in the next round. Since $|P_0\cup P_1|=10m>k$, a guard would survive. Similarly, vertices of distinct commands $c,d$ would make the $16m>k$ guards of $C_c\cup C_d$ reachable. Finally, a state vertex on side $j$ together with a transition or test vertex of command $c$ would make all guards of $P_j\cup C_c$ reachable. These sets are disjoint and have total size $5m+8m=13m>k$, so a guard would again survive. Each case contradicts the first part of the lemma. Thus every nonempty territory consists entirely of state vertices on one side or entirely of transition or test vertices of one command.
\end{proof}

\Needspace{11\baselineskip}
For a tuple $\q$ and $j\in\{0,1\}$, let $Y_c^j(\q)$ be the set of vertices of command $c$ adjacent to $X^j(\q)$, and put $Y^j(\q)=\bigcup_{c\in\cmd}Y_c^j(\q)$. These are transition vertices when $c\in\Sigma$ and test vertices when $c=\test$. For $c\in\Sigma$, the construction gives
\begin{gather*}
Y_c^0(\q)=Y_c^1\bigl(\delta_c(\q)\bigr)
\quad\text{if $\delta_c(\q)$ is defined},\\
Y_c^1(\q)=Y_c^0\bigl(\delta_{\inv c}(\q)\bigr)
\quad\text{if $\delta_{\inv c}(\q)$ is defined}.
\end{gather*}
The same transition vertices are viewed from their two sides: their command remains $c$, while traversal from side~$1$ to side~$0$ reads $\inv c$. The vertices in $Y_c^j(\q)$ are distinct for different automata, and none has two state neighbors on the same side. Hence
\begin{equation}\label{eq:commandsizes}
|Y_c^j(\q)|=m\quad\text{for every }c\in\cmd,\qquad |Y^j(\q)|=5m.
\end{equation}

\begin{lemma}\label{lem:simulation}
Suppose a winning strategy using at most $k$ hunters has $R_t=X^j(\q)$. Then $R_{t+1}=Y_c^j(\q)$ for one command $c\in\cmd$, and this set contains no bad transition vertex or bad test vertex. If $c\in\Sigma$, the territory $R_{t+2}$ is either $X^j(\q)$ or $X^{1-j}(\q')$, where
\[
\q'=\begin{cases}\delta_c(\q),&j=0,\\ \delta_{\inv c}(\q),&j=1.\end{cases}
\]
If $c=\test$, the tuple $\q$ is accepting and the hunters can capture at round $t+2$.
\end{lemma}
\begin{proof}
The neighborhood of the tuple is
\begin{equation}\label{eq:statenb}
N(X^j(\q))=P_j\cup Y^j(\q).
\end{equation}
By Lemma~\ref{lem:forbidden}, the hunters must shoot all $5m$ guards of $P_j$. At most $4m$ shots remain for the $5m$ transition and test vertices, leaving at least $m$ of them. They all have one command $c$, by the same lemma. Since each command class has size $m$, the entire class survives and all other classes are shot. Thus
\begin{equation}\label{eq:select}
W_{t+1}=P_j\cup\bigl(Y^j(\q)\setminus Y_c^j(\q)\bigr),\qquad R_{t+1}=Y_c^j(\q).
\end{equation}
This uses exactly $9m$ shots. By Lemma~\ref{lem:forbidden}, all vertices in the surviving territory $Y_c^j(\q)$ are good.

If $c\in\Sigma$, the represented transition is consequently defined in every component. Its resulting tuple $\q'$ satisfies
\begin{equation}\label{eq:transitionnb}
N(Y_c^j(\q))=C_c\cup X^j(\q)\cup X^{1-j}(\q').
\end{equation}
The two tuple sets each contain $m$ vertices and are disjoint, since they lie on different sides. All $8m$ guards of $C_c$ must be shot, leaving at most $m$ shots for the $2m$ states. At least $m$ states survive, and they must all lie on one side by Lemma~\ref{lem:forbidden}. Thus one whole tuple is shot and the other whole tuple survives.

If $c=\test$, all vertices in $Y_{\test}^j(\q)$ must be good, so $q_i\in F_i$ for every $i$. Their neighborhood is $C_{\test}\cup X^j(\q)$, of size $8m+m=k$, and shooting this set captures the rabbit at round $t+2$.
\end{proof}

For $c\in\Sigma$, the territory may move from $X^j(\q)$ to $Y_c^j(\q)$ and then return to $X^j(\q)$. In this case, each possible rabbit walk traverses the same edge in both directions. These two rounds leave the tuple unchanged and simply delay the simulation. To simulate a word $w=d_1\cdots d_\ell$, we start on side~$0$ and change sides after each letter. At odd positions $r$, we move from side~$0$ to side~$1$ using command $d_r$. At even positions $r$, we move from side~$1$ to side~$0$ using command $\inv{d_r}$: traversing the corresponding transition paths backwards reads $d_r$.

\Needspace{12\baselineskip}
For example, consider two inverse automata with transitions
\[
\begin{array}{ll}
A_1:&q_0\xrightarrow{a}q_1\xrightarrow{b}q_2\xrightarrow{a}q_3,\\[3pt]
A_2:&r_0\xrightarrow{a}r_1\xrightarrow{b}r_2\xrightarrow{a}r_0,
\end{array}
\]
together with their inverse transitions; all other transitions are undefined. Their initial states are $q_0$ and $r_0$, and their accepting sets are $F_1=\{q_3\}$ and $F_2=\{r_0\}$. Both accept $aba$, with successive tuples $\q^0=(q_0,r_0)$, $\q^1=(q_1,r_1)$, $\q^2=(q_2,r_2)$ and $\q^3=(q_3,r_0)$. Here $m=2$ and $k=18$. The three letters are simulated by commands $a,\inv b,a$, followed by the test command $\test$. Table~\ref{tab:example} gives the resulting winning strategy, using $18$ hunters in every round.

\begin{table}[htbp]
\centering\small
\renewcommand{\arraystretch}{1.18}
\begin{tabular}{@{}c l l@{}}
\toprule
Round $t$ & Shot set $W_t$ & Territory $R_t$ after the shot\\
\midrule
1 & $Z$ & $X^0(\q^0)$\\
2 & $P_0\cup\bigl(Y^0(\q^0)\setminus Y_a^0(\q^0)\bigr)$ & $Y_a^0(\q^0)$\\
3 & $C_a\cup X^0(\q^0)$ & $X^1(\q^1)$\\
4 & $P_1\cup\bigl(Y^1(\q^1)\setminus Y_{\inv b}^1(\q^1)\bigr)$ & $Y_{\inv b}^1(\q^1)$\\
5 & $C_{\inv b}\cup X^1(\q^1)$ & $X^0(\q^2)$\\
6 & $P_0\cup\bigl(Y^0(\q^2)\setminus Y_a^0(\q^2)\bigr)$ & $Y_a^0(\q^2)$\\
7 & $C_a\cup X^0(\q^2)$ & $X^1(\q^3)$\\
8 & $P_1\cup\bigl(Y^1(\q^3)\setminus Y_{\test}^1(\q^3)\bigr)$ & $Y_{\test}^1(\q^3)$\\
9 & $C_{\test}\cup X^1(\q^3)$ & $\varnothing$\\
\bottomrule
\end{tabular}
\caption{A winning strategy for the common word $aba$.}
\label{tab:example}
\end{table}

\begin{proposition}\label{prop:reduction}
The construction has polynomial size and satisfies
\begin{equation}\label{eq:reduction}
h_S(G)\leq9m\quad\Longleftrightarrow\quad\bigcap_{i=1}^{m}L(A_i)\neq\varnothing.
\end{equation}
\end{proposition}
\begin{proof}
Every winning strategy with at most $k=9m$ hunters must shoot all of $Z$ in the first round, by Lemma~\ref{lem:forbidden}. Since $|Z|=k$, this forces $W_1=Z$ and $R_1=X^0(\s)$.

Suppose a word $w=d_1\cdots d_\ell$ is accepted by all automata, and let $\q^r=\delta_{d_r}(\q^{r-1})$ for $r=1,\ldots,\ell$, with $\q^0=\s$. Start with $W_1=Z$. At step $r$, let $j=(r-1)\bmod2$ and choose $c_r=d_r$ if $j=0$, and $c_r=\inv{d_r}$ if $j=1$. Set
\[
\begin{aligned}
W_{2r}&=P_j\cup\bigl(Y^j(\q^{r-1})\setminus Y_{c_r}^j(\q^{r-1})\bigr),\\
W_{2r+1}&=C_{c_r}\cup X^j(\q^{r-1}).
\end{aligned}
\]
By~\eqref{eq:transitionnb}, these shots leave $Y_{c_r}^j(\q^{r-1})$ and then $X^{1-j}(\q^r)$. In particular, when $j=1$, the effective letter is $\inv{c_r}=d_r$. Each shot uses exactly $k$ hunters. After all letters, select the test command and shoot $C_{\test}$ together with the accepting tuple. This wins in $2\ell+3$ rounds. For the empty word, apply the test immediately after the first round.

Conversely, suppose the intersection is empty and a winning strategy with at most $k$ hunters exists. After the forced first shot, we have $R_1=X^0(\s)$. If $R_t=X^j(\q)$ and $\q$ is reachable from $\s$ by a common word, then $\q$ is not accepting. Lemma~\ref{lem:simulation} therefore excludes the test command: the next pair of rounds either returns to the same tuple or applies a common letter and changes sides. The resulting tuple is again reachable. Inductively, every odd-round territory contains $m$ states and every even-round territory contains $m$ transition vertices. The territory never becomes empty, a contradiction.

For the size bound, let $n=\sum_i|Q_i|$. There are $2n$ state vertices. For fixed $i,c$, a partial injection with $t$ defined transitions contributes $t$ good transition vertices and $2(|Q_i|-t)$ bad transition vertices. The four transition commands therefore contribute at most $8n$ transition vertices, and the test command contributes $2n$ test vertices. Including the guards, $Z$, and $U$ gives
\[
|V(G)|\leq 12n+68m+2.
\]
Since $m\leq n$, the construction has $O(n)$ vertices and $O(n^2)$ edges and is computable in polynomial time.
\end{proof}

\subsection{Unrestricted and directed games}
\begin{proof}[Proof of Theorem~\ref{thm:main}]
Membership in PSPACE follows from Lemma~\ref{lem:space}. For hardness, apply Proposition~\ref{prop:reduction}, and write $N=|V(G)|$ and $k=9m$. We have $1\leq k\leq|S|$. The construction in the proof of~\cite[Proposition~4.2]{BGM} adds to $G$ three looped cliques $A,B,C$ of sizes $N-k,k,2k$, respectively, completely joins $A$ to $B,C,S$, and completely joins $B$ to $V(G)$. Figure~\ref{fig:finalreductions} illustrates this construction. Since $A\cup C$ induces a complete graph with a loop at every vertex and has $(N-k)+2k=N+k$ vertices, $h(H)\geq N+k$. The equivalence proved in~\cite[Proposition~4.2]{BGM} therefore gives
\begin{equation}\label{eq:initialtransform}
h_S(G)\leq k\quad\Longleftrightarrow\quad h(H)=N+k\quad\Longleftrightarrow\quad h(H)\leq N+k.
\end{equation}
Construct a graph $B_H$ by replacing each vertex $v$ of $H$ with two vertices $v^0,v^1$, and each edge $uv$ with the edges $u^0v^1$ and $u^1v^0$. A loop at $v$ gives the single edge $v^0v^1$. By~\cite[Lemma~3.2]{BGM}, $h(B_H)=h(H)$. The graph $B_H$ is simple and bipartite, and both transformations have polynomial size. Together with~\eqref{eq:reduction} and~\eqref{eq:initialtransform}, this proves the result. We also record that $B_H$ has no isolated vertices and, writing $K=N+k$, has $4K$ vertices, with $K\geq3$.
\end{proof}

\begin{corollary}\label{cor:directed}
Deciding whether $h(D)\leq k$ is PSPACE-complete for digraphs, even for symmetric digraphs with bipartite underlying graphs. Deciding whether one hunter wins is PSPACE-complete even for oriented graphs with positive minimum outdegree.
\end{corollary}
\begin{proof}
Replacing every edge of an undirected graph by two opposite arcs preserves its hunting number. Thus Theorem~\ref{thm:main} gives the first assertion.

For the second, use the vertex-subdivision reduction of Ben-Ameur and Maddaloni~\cite[Theorem~4.3]{BM}. For a digraph $F$ with positive minimum outdegree and $1\leq b\leq|V(F)|$, replace every vertex $v$ by a directed path $v^1,\ldots,v^b$, and replace each arc $uv$ by $u^bv^1$. Denote the resulting digraph by $F^{[b]}$. The cited reduction gives
\begin{equation}\label{eq:subdivision}
h(F)\leq b\quad\Longleftrightarrow\quad h(F^{[b]})\leq1.
\end{equation}
Apply it to the symmetric digraph $F$ obtained from $B_H$ in the preceding proof, with $b=K$. Here $b\geq3$ and $|V(F)|=4b$, so the construction is polynomial. Every arc of $F^{[b]}$ goes from one level to the next, cyclically through the $b$ levels. Since $b\geq3$, there are no loops or oppositely directed pairs of arcs; thus $F^{[b]}$ is oriented. Its minimum outdegree is positive, so $h(F^{[b]})\geq1$. Equation~\eqref{eq:subdivision} proves hardness for one hunter. Membership in PSPACE follows from Lemma~\ref{lem:space}.
\end{proof}

\begin{figure}[htbp]
\centering
\resizebox{.9\linewidth}{!}{%
\begin{tikzpicture}[x=1cm,y=1cm,font=\small,
 block/.style={draw,rounded corners=2pt,minimum width=1.9cm,minimum height=1.0cm,align=center,fill=white},
 vert/.style={circle,draw,fill=white,minimum size=7mm,inner sep=1pt},
 internal/.style={-{Stealth[length=2mm]},line width=.8pt,draw=stateblue},
 inherited/.style={-{Stealth[length=2mm]},line width=.8pt,draw=badred,dashed},
 fulljoin/.style={line width=1.1pt},
 lab/.style={anchor=west,font=\small\bfseries}]
\path[use as bounding box](-.2,-5.1) rectangle(10.4,.35);
\node[lab] at (0,0) {Removing the restriction to the initial set $S$};
\fill[testgreen!7,rounded corners=7pt] (.0,-.55) rectangle (3.4,-4.05);
\node[block,draw=testgreen] (C) at (1.7,-1.25) {$C$\\$K^{\circ}_{2k}$};
\node[block,draw=testgreen] (A) at (1.7,-3.35) {$A$\\$K^{\circ}_{N-k}$};
\node[block] (B) at (6.2,-1.25) {$B$\\$K^{\circ}_{k}$};
\draw[fulljoin,draw=testgreen] (A)--(C);
\draw[fulljoin] (A.north east)--(B.south west);
\draw[rounded corners=4pt,fill=stateblue!3] (4.8,-2.4) rectangle (10.2,-4.1);
\node[font=\scriptsize,anchor=east] at (10.05,-3.90) {$G$, $|V(G)|=N$};
\node[block,minimum height=.7cm,minimum width=1.35cm,fill=white] (S) at (6.2,-3.35) {$S$};
\node[block,minimum height=.7cm,minimum width=1.65cm,fill=white] (rest) at (8.85,-3.35) {$V(G)\setminus S$};
\draw[fulljoin] (A.east)--(S.west);
\draw[fulljoin] (B.south)--(S.north);
\draw[fulljoin] (B.east) to[out=0,in=90] (rest.north);
\node[align=center,text=testgreen!70!black,font=\small] at (5.1,-4.7) {The shaded sets induce $H[A\cup C]=K^{\circ}_{N+k}$,\\so $h(H)\geq N+k$.};

\end{tikzpicture}%
}
\caption{Removing the restriction to the initial set $S$, using the construction of Ben-Ameur, Gahlawat and Maddaloni~\cite[Proposition~4.2]{BGM}. Thick lines represent complete joins, and all added joins are shown. Edges within $G$ are retained but not drawn. The notation $K^{\circ}_t$ denotes a clique of order $t$ with a loop at each vertex. The shaded sets induce a looped clique of order $N+k$.}
\label{fig:finalreductions}
\end{figure}

\subsection{Firefighting}
In the firefighting game of Althoetmar, Schade and Sch\"urenberg~\cite{Firefighting}, all vertices are initially burning. At each round, at most $k$ vertices are extinguished, after which fire spreads from every remaining burning vertex to its neighbors; the remaining burning vertices also stay burning. The goal is to extinguish the whole graph. Write $\operatorname{ffn}(G)$ for the minimum required number of firefighters. Equivalently, the fugitive in the Hunters and Rabbit game may remain at its current vertex, so $\operatorname{ffn}(G)$ is also the inspection number of Bernshteyn and Lee~\cite{BernshteynLee}.

To match the timing convention in~\eqref{eq:territory}, let $R_t$ denote the burning vertices immediately after extinguishing in round $t$, before fire spreads. Then
\[
R_1=V(G)\setminus W_1,\qquad R_{t+1}=N_G[R_t]\setminus W_{t+1}.
\]
The burning set after the subsequent spread is $N_G[R_t]$, which is empty exactly when $R_t$ is empty. Thus, if $G^\circ$ is obtained from $G$ by adding a loop at every vertex, then $\operatorname{ffn}(G)=h(G^\circ)$.

\begin{theorem}\label{thm:firefighting}
Deciding whether $\operatorname{ffn}(G)\leq k$ is PSPACE-complete for simple undirected graphs.
\end{theorem}
\begin{proof}
Membership in PSPACE follows from Lemma~\ref{lem:space} applied to $G^\circ$.

For hardness, we adapt the inverse-automata simulation. Loops introduce a new possibility: a territory can retain vertices of its current class while acquiring vertices of a neighboring class. The guard sizes must therefore also rule out territories mixing state vertices with transition or test vertices. We choose them so that every state tuple and every good transition class has a neighborhood of size $k+m$, whereas a good test class has a neighborhood of size exactly $k$.

Use the same state, transition and test vertices and the same adjacencies as in the construction of $G$, with budget $k=11m$ and guard sizes $|P_0|=|P_1|=6m$ and $|C_c|=9m$ for every $c\in\cmd$. Set $|Z|=k$ and $|U|=k+2$, and add a loop at every vertex. Denote this graph by $\widetilde G$ and put $S=Z\cup X^0(\s)$.

The argument of Lemma~\ref{lem:forbidden} excludes vertices of $U$, guards, bad transition vertices, bad test vertices and vertices of $Z$ from every territory of a winning strategy using at most $k$ hunters. States on both sides would expose $12m>k$ guards, and vertices of two distinct commands would expose $18m>k$ guards. A state and a transition or test vertex would expose $|P_j|+|C_c|=15m>k$ guards. Thus each nonempty territory consists entirely of states on one side or entirely of transition or test vertices of one command.

Suppose the current territory is $X=X^j(\q)$. The loops give $N_{\widetilde G}(X)=P_j\cup X\cup Y^j(\q)$. All $6m$ guards must be shot, leaving at most $5m$ shots for the $6m$ vertices in $X\cup Y^j(\q)$. At least $m$ of these vertices survive. The preceding restrictions force the next territory to be either the whole set $X$ or one whole command class $Y_c^j(\q)$, whose vertices must all be good.

Now let $Y=Y_c^j(\q)$ be a good transition class, with $c\in\Sigma$, and put $X'=X^{1-j}(\q')$, where $\q'=\delta_c(\q)$ if $j=0$ and $\q'=\delta_{\inv c}(\q)$ if $j=1$. We have $N_{\widetilde G}(Y)=C_c\cup Y\cup X\cup X'$. Shooting the $9m$ guards leaves at most $2m$ shots for three disjoint sets of size $m$. At least $m$ vertices survive, and they must form exactly one of $Y$, $X$ and $X'$.

If $Y=Y_{\test}^j(\q)$ is a good test class, then every $q_i$ belongs to $F_i$. Its neighborhood is $C_{\test}\cup Y\cup X$, of size $9m+m+m=k$, so the hunters can capture in the next round.

The first shot is forced to be $Z$, leaving $X^0(\s)$. Before a test class is reached, the territory can remain unchanged, return from a transition class to the same tuple, or move to the opposite tuple by applying a common letter. Inductively, every tuple reached is reachable from $\s$ by a common word. A state tuple and a good transition class each have a neighborhood of size $12m>k$, so capture requires reaching a good test class and hence an accepting tuple. Conversely, a common accepted word gives a winning strategy by following its transitions and then applying the test: at each round, shoot all vertices of the neighborhood except those of the desired next territory. Each such shot uses $k$ hunters. Consequently,
\begin{equation}\label{eq:firefighting}
h_S(\widetilde G)\leq k\quad\Longleftrightarrow\quad\bigcap_{i=1}^m L(A_i)\neq\varnothing.
\end{equation}

Finally, put $N=|V(\widetilde G)|$ and apply the construction of~\cite[Proposition~4.2]{BGM} used in the proof of Theorem~\ref{thm:main}. Its equivalence~\eqref{eq:initialtransform} gives $h(H)\leq N+k$ if and only if $h_S(\widetilde G)\leq k$. Every vertex of $\widetilde G$ has a loop, and all vertices added by this construction also have loops. Thus, if $J$ is the simple graph obtained by deleting the loops of $H$, then $\operatorname{ffn}(J)=h(H)$. Both constructions have polynomial size, which completes the reduction.
\end{proof}

\section*{Declaration of generative AI and AI-assisted technologies}
During the preparation of this work, the authors used \AItool{} to assist with proof checking, literature searches, English language editing, and the preparation of LaTeX and TikZ code. The authors reviewed and edited the resulting material and take full responsibility for the content of the final manuscript.

\end{document}